\documentclass[journal]{IEEEtran}

\usepackage{amsmath,amsfonts, amssymb,amsthm}

\newtheorem{thm}{Theorem}
\newtheorem{prop}{Proposition}

\newtheorem{lemma}{Lemma}

\def\be{\begin{equation}}
\def\en{\end{equation}}
\def\bee{\begin{eqnarray*}}
\def\ene{\end{eqnarray*}}

\def\E{{\bf E}}

\def\R{{\mathbb{R}}}
\def\E{{\mathbb{E}}}

\date{}

\vspace{2in}

\begin{document}

\title{Variants of the entropy power inequality}

\author{Sergey G. Bobkov$^1$\thanks{${}^1$School of Mathematics, 
University of Minnesota, 127 Vincent Hall, 206 Church St. S.E., Minneapolis, 
MN 55455 USA. Email: bobkov@math.umn.edu 
\newline Partially supported by the Alexander von Humboldt Foundation and 
NSF grant DMS-1612961} and 
Arnaud Marsiglietti$^2$\thanks{$^2$California Institute of Technology, 
1200 East California Boulevard, MC 305-16, Pasadena, CA 91125 USA. 
Email: amarsigl@caltech.edu 
\newline Supported by the Walter S. Baer and Jeri Weiss CMI Postdoctoral Fellowship.}}

\maketitle

\begin{abstract}

An extension of the entropy power inequality to the form 
$N_r^\alpha(X+Y) \geq N_r^\alpha(X) + N_r^\alpha(Y)$
with arbitrary independent summands $X$ and $Y$ in $\R^n$ is obtained 
for the R\'enyi entropy and powers $\alpha \geq (r+1)/2$.

\end{abstract}


\section{Introduction}\label{1}

Given a continuous random vector $X$ in $\R^n$ with density $f$, define
the (Shannon) entropy and the associated entropy power
\bee
h(X)
 & = & 
-\int_{\R^n} f(x)\log f(x)\,dx, \\
N(f)
 & = &
N(X) \ = \ \exp\Big\{\frac{2}{n}\,h(X)\Big\}.
\ene
Serving as measures of ``chaos" or ``randomness" hidden in the distribution 
of $X$, these functionals possess a number of remarkable properties, especially
when they are considered on convolutions. For example, we have the famous
entropy power inequality (EPI), fundamental in Information Theory. It states that
\be\label{EPI}
N(X+Y) \geq N(X) + N(Y),
\en
for arbitrary independent summands $X$ and $Y$ in $\R^n$ whenever the involved
entropies are well defined (cf. \cite{Sh}, \cite{St}).
Several proofs of the EPI exist 
(see e.g. \cite{L}, \cite{D-C-T}, \cite{S-V}, \cite{V-G}, \cite{R}, \cite{W-M}), 
as well as refinements (see e.g. \cite{A-B-B-N}, \cite{M-B}, \cite{Co}). 
We refer to the survey \cite{M-M-X} for further details.
Moreover, when a Gaussian noise is added to $X$, i.e., if $Y = \sqrt{t}\, Z$ 
with $Z$ standard normal, the random vector $X+\sqrt{t}\, Z$ has
density $f_t$ whose entropy power is a concave function in $t$, so that
\be\label{concave}
\frac{d^2}{dt^2}\, N(f_t) \leq 0 \qquad (t>0).
\en
This observation due to Costa \cite{C}, which strengthens \eqref{EPI} 
in the special case where $Y$ is Gaussian, is known as the concavity of 
entropy power theorem (cf. also \cite{D}, \cite{V}).

There has been large interest in extending such properties to more general 
informational functionals, in particular, to the R\'enyi entropy and 
R\'enyi entropy power
\bee
h_r(X) 
 & \!\!\!= & \!\!
-\frac{1}{r-1}\,\log \int_{\R^n} f(x)^r\,dx, \\
N_r(X) 
 & \!\!\!= & \!\!
\exp\Big\{\frac{2}{n}\,h_r(X)\Big\} = 
\bigg(\int_{\R^n} f(x)^r\,dx\bigg)^{-\frac{2}{n}\,\frac{1}{r-1}}
\ene
of a fixed order $r>0$, or by some natural functionals of $h_r$ and $N_r$. 
As one interesting example, for the densities $u(x,t) = f_t(x)$ solving 
the nonlinear heat equation $\frac{\partial}{\partial t}\, u = \Delta u^r$ 
with $r > 1 - \frac{2}{n}$, Savar\'e and Toscani \cite{S-T} have extended 
property \eqref{concave} to the functional $N_r^\alpha$ in place of $N$, where 
$\alpha = 1 + \frac{n}{2}\,(r-1)$. Therefore, in this PDE context, 
it is natural to work with
$$
\widetilde N_r(X) = 
\bigg(\int_{\R^n} f(x)^r\,dx\bigg)^{-\frac{2}{n}\,\frac{1}{r-1} - 1},
$$
called the $r$-th R\'enyi power in \cite{S-T}. Although the solutions $f_t$
lose the convolution structure, one may wonder whether or not
the Savar\'e-Toscani entropy power $\widetilde N_r$ shares the EPI \eqref{EPI} 
as well. Here we give an affirmative answer to this question, including sharper 
powers of $N_r$.

\begin{thm}\label{main-thm}
Given independent random vectors $X$ and $Y$ in $\R^n$ with densities, we have
\be\label{main}
N_r^\alpha(X+Y) \geq N_r^\alpha(X) + N_r^\alpha(Y)
\en
whenever $\alpha \geq \frac{r+1}{2}$ $(r > 1)$.
\end{thm}

Letting $r \downarrow 1$, inequality \eqref{main} returns us to the classical EPI.
This inequality is getting sharper when $r$ is fixed and $\alpha$ decreases. Anyhow, \eqref{main} is no longer true for $\alpha = 1$ like in \eqref{EPI}.
For the range $r > 3$, this fact was mentioned in \cite{B-C} in case where 
both $X$ and $Y$ are uniformly distributed. As we will see,
\eqref{main} may be violated with $\alpha = 1$ for any $r>1$, even when one of 
the summands is normally distributed (that is, for the densities $f_t$ 
in the heat semigroup model).

For $r=\infty$, a R\'enyi entropy power inequality of the form \eqref{main} cannot 
hold, for any $\alpha$. Indeed, if we take $X$ and $Y$ uniformly distributed 
on $[0,1]$, then $N_{\infty} (X+Y) = N_{\infty}(X)$. We refer to 
\cite{B-C2}, \cite{M-M-X2} for recent developments on $N_{\infty}$. 
While there has been several results about the R\'enyi entropy power of order 
$r \geq 1$, the investigation of a R\'enyi entropy power inequality for 
the R\'enyi entropy of order $r < 1$ has been addressed only very recently (see \cite{MM}).

In the proof of \eqref{main} we follow an approach of Lieb \cite{L}, employing 
Young's inequality with best constants. Although the basic argument is rather 
standard, we recall it in the next section. In our situation it leads to some 
routine calculus computations, so we move the involved analysis to separate 
sections (starting with the case of equal entropy powers). In Section \ref{5}, 
we analyze \eqref{main} with $\alpha = 1$ and show that this inequality cannot 
be true in general. In Section \ref{6} we provide a simple lower bound
on the optimal exponent $\alpha = \alpha(r)$ in \eqref{main}.
Finally, in Section \ref{7}, we conclude with remarks on the monotonicity
of R\'enyi's entropy along rescaled convolutions.

\section{Information-theoretic formulation of Young's inequality}\label{2}

The Young inequality with optimal constants (due to Beckner \cite{B1} and
Brascamp and Lieb \cite{B-L}) indicates that, for any two independent random 
vectors $X$ and $Y$ in $\R^n$ with densities $f$ and $g$, respectively, 
and for all parameters $p,q,r \geq 1$ such that
\be\label{conjugate}
\frac{1}{p'} + \frac{1}{q'} = \frac{1}{r'},
\en
we have
\be\label{sharp-young}
\|f * g\|_r \,\leq\, C^{\frac{n}{2}}\, \|f\|_p\, \|g\|_q
\en
with
\be\label{constant}
C = C(p,q,r) = \frac{c_p c_q}{c_r}, \qquad {\rm where} \ \ 
c_\alpha = \frac{\alpha^{1/\alpha}}{(\alpha')^{1/\alpha'}}.
\en
As usual, $f * g$ denotes the convolution,
$p' = \frac{p}{p-1}$ is the conjugate power, and
$$
\|f\|_p = \bigg(\int_{\R^n} f(x)^p\,dx\bigg)^{1/p}
$$
stands for the $L^p$-norm of a non-negative function $f$ on $\R^n$ with respect 
to the Lebesgue measure. In general, we have $C \leq 1$, with equality in 
\eqref{sharp-young} attainable for Gaussian densities (the traditional Young 
inequality is formulated without this constant, so, in a weaker form).

Since $\|f\|_r = N_r(X)^{-\frac{n}{2r'}}$, the inequality \eqref{sharp-young} 
may be stated as a dimension-free relation between the corresponding entropy 
powers, namely
\be\label{dimension-free}
N_r(X+Y)^{\frac{1}{r'}} \geq 
\frac{1}{C}\,N_p(X)^{\frac{1}{p'}} N_q(Y)^{\frac{1}{q'}}.
\en
This is an equivalent information-theoretic formulation of Beckner's result, 
specialized to the class of probability densities,
which appears, for example, in the book by Cover and Thomas \cite{C-T}
(in a slightly different form, cf. Theorem 17.8.3, p. 677).

It is natural to have an analog of \eqref{dimension-free} for one functional 
$N_r$ only (rather than for three parameters). This can be done on the basis 
of \eqref{sharp-young} by noting that, due to Jensen's (or H\"older's) inequality, 
and since $p,q \leq r$ in \eqref{conjugate} and $f$ is a probability density 
function,
$$
\|f\|_p^p \, \leq \, 
\|f\|_1^{1 - \frac{p-1}{r-1}} \|f\|_r^{r\frac{p-1}{r-1}} \, = \,
\|f\|_r^{r\frac{p-1}{r-1}} \qquad (r>1).
$$
As an alternative approach, one can just use the monotonicity of the function
$r \rightarrow N_r$, which follows, for example, from the representation
$$
N_r^{-\frac{n}{2}}(X) \, = \, \big[\,\E\,f(X)^{r-1}\big]^{\frac{1}{r-1}}.
$$
Hence $N_p \geq N_r$, $N_q \geq N_r$ in \eqref{dimension-free}, and 
with these bounds it immediately yields:

\begin{prop}\label{young}
Given independent random vectors $X$ and $Y$ in 
$\R^n$ with densities, we have
\be\label{eq:young}
N_r(X+Y)^{\frac{1}{r'}} \, \geq \,
\frac{1}{C}\,N_r(X)^{\frac{1}{p'}}\, N_r(Y)^{\frac{1}{q'}},
\en
which holds true for all $p,q,r \geq 1$ subject to \eqref{conjugate}
with constant $C = C(p,q,r)$ as in \eqref{constant}.
\end{prop}

A weak point of this inequality is however the loss of equality for 
Gaussian densities. Nevertheless, there is still freedom to optimize
the right-hand side over all admissible couples $(p,q)$, or to choose 
specific values, even if they are not optimal.

Notice that by Jensen's inequality, we always have
$$ 
N_r(X+Y) \geq \max\{N_r(X), N_r(Y)\}, 
$$
hence inequality \eqref{main} trivially holds if $N_r(X)N_r(Y)=0$. 
Therefore, one may assume without loss of generality that $N_r(X)N_r(Y)>0$, 
and we will implicitly make this assumption in the next sections.

\section{The case of equal entropy powers}\label{3}

Let us illustrate this approach in the simpler situation of equal R\'enyi 
entropies. When $N_r(X) = N_r(Y) = N$, inequality \eqref{eq:young} 
is simplified to
\be\label{simple}
N_r(X+Y) \geq C^{-r'} N, 
\en
and our task reduces to the minimization of $C$ 
as a function of $(p,q)$ for a fixed $r > 1$. 
Putting $x = \frac{1}{p'}$, $y = \frac{1}{q'}$, so that
$\frac{1}{p} = 1 - \frac{1}{p'} = 1-x$ and 
$\frac{1}{q} = 1-\frac{1}{q'} = 1-y$, 
from \eqref{constant},
\begin{equation}\label{C}
\begin{aligned}
\frac{1}{C} \, = \, \frac{c_r}{c_p c_q} \, = & \,
c_r\, \frac{(\frac{1}{p})^{1/p}}{(\frac{1}{p'})^{1/p'}}\,  
\frac{(\frac{1}{q})^{1/q}}{(\frac{1}{q'})^{1/q'}} \, \\ = & \,
c_r\,\frac{(1-x)^{1-x}}{x^x} \,\frac{(1-y)^{1-y}}{y^y}.
\end{aligned}
\end{equation}
Hence, we need to maximize the quantity
\bee
\psi(x) 
 & = &
\log \frac{1}{C} \\
 & = &
\log c_r - \big(x\log x - (1-x)\log(1-x)\big) \\ 
 & ~ & - \big(y\log y - (1-y)\log(1-y)\big) 
\ene
subject to the constraint \eqref{conjugate}, that is, for $x,y \geq 0$, 
$x+y = \frac{1}{r'}$, or equivalently, on the interval $0 \leq  x \leq \frac{1}{r'}$.
At the endpoints, we have $\psi(0) = \psi(1/r') = 0$, while inside the interval
$$
\psi'(x) = \log \frac{y(1-y)}{x(1-x)} \, = \, 0
$$
if and only if $ y(1-y) = x(1-x)$. This equation is solved either as
$y = x = \frac{1}{2r'}$ or as $y=1-x$. But the latter contradicts 
$x+y = \frac{1}{r'} < 1$. Moreover, 
$$
\psi''(x) = \frac{2x-1}{x(1-x)} + \frac{2y-1}{y(1-y)}
$$
is negative at $x_0=\frac{1}{2r'}$, which implies that $x_0$ 
is the point of maximum of the function $\psi$.

Thus, the coefficient $\frac{1}{C}$ in \eqref{eq:young} is maximized, when 
$p' = q' = 2r' = \frac{2r}{r-1}$. For these values,
$p = q = \frac{2r}{r+1}$,
so
$$ 
c_p = c_q = (2r)^{\frac{1}{r}}\, (r-1)^{\frac{r-1}{2r}}\,(r+1)^{-\frac{r+1}{2r}}, 
$$
and
$$ 
C \, = \, 2^{\frac{2}{r}}\,r\,(r+1)^{-\frac{r+1}{r}}. 
$$
It remains to raise $C$ to the power $-r' = - \frac{r}{r-1}$, and then 
we obtain an explicit expression for the optimal constant in \eqref{simple} 
derived on the basis of \eqref{eq:young}.

\begin{prop}\label{equal}
If the independent random vectors $X$ and $Y$ 
satisfy $N_r(X) = N_r(Y) = N$ for some $r \geq 1$, then
\be\label{eq:equal}
N_r(X+Y) \, \geq \, A_r N, \qquad 
A_r = 4^{-\frac{1}{r-1}}\,(r+1)^{\frac{r+1}{r-1}}\, r^{-\frac{r}{r-1}}.
\en
\end{prop}

Proposition \ref{equal} is not new and a more general version where 
the distributions have different R\'enyi entropies was obtained in 
\cite[Theorem 1]{R-S}. Moreover, in the case of different R\'enyi entropies, 
tighter bounds are provided in \cite[Corollary 3]{R-S}.

Now, it is easy to see that $1 < A_r < 2$. Moreover, \eqref{eq:equal} 
provides the desired linear bound \eqref{main} in case of equal entropy powers, 
$
N_r^\alpha(X+Y) \geq 2 N^\alpha(X),
$
as long as $A_r \geq 2^{1/\alpha}$, or equivalently, when
$\alpha > \alpha(r) = (\log 2)/(\log A_r)$.
A simple analysis shows that $\alpha(r) \leq (r+1)/2$.

\section{The general case}\label{4}

Here we derive the extension \eqref{main} of the EPI for the power 
$\alpha = \frac{r+1}{2}$ in the case of arbitrary values $N_r(X)$ and $N_r(Y)$. 
As a preliminary step, let us return to the inequality \eqref{eq:young} 
and raise it to the power $\alpha$, so as to rewrite it as
$$
N_r^\alpha(X+Y) \geq 
C^{-\alpha r'}\,N_r(X)^{\frac{\alpha r'}{p'}} N_r(Y)^{\frac{\alpha r'}{q'}}.
$$
Putting $x = N_r^\alpha(X)$, $y = N_r^\alpha(Y)$ and assuming without loss 
of generality that $x+y=\frac{1}{r'}$ (using homogeneity of these functionals), 
it is enough to show that
$$
C^{-\alpha r'}\,x^{\frac{r'}{p'}} y^{\frac{r'}{q'}} \geq \frac{1}{r'}
$$
for some admissible $p,q \geq 1$, i.e., satisfying the condition
$\frac{1}{p'} + \frac{1}{q'} = \frac{1}{r'}$. Hence, Theorem \ref{main-thm} 
will immediately follow from the following lemma.

\begin{lemma}\label{cle}
Let $r>1$. Let $x,y > 0$ be such that $x+y = \frac{1}{r'}$. Then, there exist 
$p,q \geq 1$ satisfying $\frac{1}{p'} + \frac{1}{q'} = \frac{1}{r'}$ such that
$$
C^{-\alpha r'}\,x^{\frac{r'}{p'}} y^{\frac{r'}{q'}} \geq \frac{1}{r'},
$$
where $C=C(p,q,r)$ is as in \eqref{constant}, and $\alpha = \frac{r+1}{2}$.
\end{lemma}

To prove Lemma \ref{cle}, we make use of the following calculus lemma.

\begin{lemma}\label{key}
Given $0 < c < 1$ and $\beta \geq \frac{2}{c} - 1$, the function
$$
\psi(x) = \frac{(1 - x)^{\beta(1 - x)}\, (1 - y)^{\beta(1 - y)}}{x^x y^y} 
\quad (y=c-x)
$$
attains minimum on the interval $0 \leq x \leq c$
either at the endpoints $x=0$, $x=c$, or at the center 
$x = \frac{c}{2}$. Moreover, in case $\beta = \frac{2}{c} - 1$, 
this function attains minimum at the endpoints.
\end{lemma}

\begin{proof}
Inside the interval $(0,c)$ the function 
\bee
v(x) 
 & = &
\log\, \psi(x) \\
 & = &
\beta ((1-x)\log(1-x) + (1-y)\log(1-y)) \\ & ~ & - (x\log x + y\log y)
\ene
has the first two derivatives
$$
v'(x) = -\beta\,\big(\log(1-x) - \log(1-y)\big) - (\log x - \log y),
$$
\bee
v''(x) 
 & = &
\Big(\frac{\beta}{1-x} + \frac{\beta}{1-y}\Big) - 
\Big(\frac{1}{x} + \frac{1}{y}\Big) \\
 & = &
\frac{\beta (2-c)}{(1-x)(1-y)} - \frac{c}{xy}.
\ene
Note that $v(0) = v(c)$. Also, $v'(0+) = \infty$, $v'(c-) = -\infty$, 
so $v$ is increasing near zero and is decreasing near the point $c$. 
In addition, $v''(x)$ is vanishing, if and only if
\be\label{4.4}
w(x) \equiv \beta (2-c)\, xy - c\,(1-x)(1-y) = 0,
\en 
which is a quadratic equation (recall that $y=c-x$).
In general it has at most two roots.

\vskip2mm
Case 1: Equation \eqref{4.4} has at most one root in $(0,c)$. Since $w(0) < 0$, 
it means that $w(x) \leq 0$ in $(0,c)$. Therefore, $v$ is concave, and thus 
attains its minimum at the endpoints of this interval.

\vskip2mm
Case 2: Equation \eqref{4.4} has exactly two roots in $(0,c)$, say 
$0 < x_1 < x_2 < c$. Since $w(0) < 0$ and $w(c) < 0$, it means that
$w(x)< 0$ in $(0,x_1)$ and $(x_2,c)$, while $w(x) > 0$ in $(x_1,x_2)$. 
That is, $v$ is strictly concave on $(0,x_1)$ and $(x_2,c)$, and is 
strictly convex on the intermediate interval. Hence, in this case there is 
at most one point of local minimum. If there is no point of local minimum, 
then $v$ attains its minimum at the endpoints. It there is 
one point $x_0$ of local minimum of $v$, then it must belong to $(x_1,x_2)$, 
and there are two points of local maximum, say $z_1$ and $z_2$ belonging 
to the other subintervals. In particular, $v' < 0$ on $(z_1,x_0)$ and $v'>0$ 
on $(x_0,z_2)$.

Note that $v'(c/2) = 0$, so this point is a candidate for local extremum. 
Moreover, by the assumption on $\beta$,
$$
v''(c/2) = \frac{4\,(\beta c - (2-c))}{c(2-c)} \geq 0.
$$
If $\beta > \frac{2}{c} - 1$, then $v''(c/2) > 0$ which means that $x_0 = c/2$ 
is a local minimum for $v$ and therefore for $\psi$, and the first assertion 
follows. If $\beta = \frac{2}{c} - 1$, then $v''(c/2) = 0$ which means that 
either $c/2 = x_1$ or $c/2 = x_2$. But at these points the derivative of $v$ 
may not vanish. In other words, the equality
$\beta = \frac{2}{c} - 1$ is only possible under Case 1. 
\end{proof}

\begin{proof}[Proof of Lemma \ref{cle}]
The best values of $p$ and $q$ can be described implicitly as solutions 
to a certain equation, and we prefer to take some specific values. 
As a natural choice, consider $(p,q)$ such that $\frac{1}{p'} = x$ and 
$\frac{1}{q'} = y$ and try to check the desired inequality 
$x^{xr'} y^{yr'} \geq C^{\alpha r'}\frac{1}{r'}$, i.e.,
$$
x^{x} y^{y} \geq C^{\alpha}\frac{1}{(r')^{1/r'}} \qquad 
\Big(x,y > 0, \ x+y=\frac{1}{r'}\Big).
$$
Equivalently, so that to eliminate the parameter $r$, we need to check 
whether or not
\be\label{no-r}
x^{x} y^{y} \geq C^{\alpha} (x+y)^{x+y} \qquad (x,y > 0, \ x+y < 1).
\en
As in Section \ref{3}, cf. \eqref{C},
\bee
C 
 & = &
\frac{c_p c_q}{c_r} \\
 & = &
\frac{x^x}{(1 - x)^{1 - x}} \ \frac{y^y}{(1 - y)^{1 - y}} \ 
\frac{(1 - x - y)^{1 - x - y}}{(x+y)^{x+y}},
\ene
and \eqref{no-r} takes the form
\be\label{4.2}
\bigg(\frac{(x+y)^{x+y}}{x^{x} y^{y}}\bigg)^{\alpha - 1} \, \geq \, 
\bigg(\frac{(1-x-y)^{1-x-y}}{(1-x)^{1-x}\, (1-y)^{1-y}}\bigg)^\alpha,
\en
or equivalently
\be\label{4.5}
\frac{(1 - x)^{\beta(1 - x)}\, (1 - y)^{\beta(1 - y)}}{x^x y^y} \, \geq \, 
\frac{(1 - x - y)^{\beta(1 - x - y)}}{(x+y)^{x+y}},
\en
where
$$ \beta = \frac{\alpha}{\alpha - 1}. $$
Here the right-hand side depends only on $c = x+y$ (since $\alpha$ may only 
depend on $r$ which is a function of $x+y$). Hence, to prove \eqref{4.5}, it is 
sufficient to minimize the left-hand side under the constraint $x,y \geq 0$, 
$x+y=c$, and then to compare the minimum with the right-hand side.
In case $\alpha = \frac{r+1}{2}$, we have
$$
\frac{\alpha}{\alpha - 1} = \frac{2 - x - y}{x+y} = \frac{2-c}{c} = 
\frac{2}{c} - 1,
$$
which is exactly the extreme value for $\beta$ in Lemma \ref{key}. Therefore, by 
its conclusion, the left-hand side of \eqref{4.5} is minimized either at $x=0$ 
or $x=c$. But for such boundary values there is equality in \eqref{4.5}. 
As a result, we obtain the desired inequality \eqref{no-r} for all $x,y > 0$ 
such that $x+y < 1$.
\end{proof}

\section{R\'enyi entropy powers for the heat semi-group}\label{5}

Let us now look at the possible behavior of the R\'enyi entropy powers 
in the class of densities $f_t$ of $X_t = X + \sqrt{t} Z$, assuming that 
$X$ has a sufficiently regular positive density $f$ (on the line), and $Z$ 
is a standard normal random variable independent of $X$.
Since for small $t>0$
$$
f_t(x) = f(x) + \frac{1}{2}\,f''(x)\,t + \mbox{o}(t),
$$
we find, by Taylor expansion and integrating by parts,
\bee
\int_{-\infty}^\infty f_t(x)^r\,dx 
 & = &
\int_{-\infty}^\infty f(x)^r \, dx \\ 
 & & \hskip-15mm - \ 
\frac{t}{2}\,r(r-1) 
\int_{-\infty}^\infty f(x)^{r-2} f'(x)^2 \, dx + \mbox{o}(t),
\ene
and thus, for $r > 1$,
\bee
N_r(X_t) 
 & = & 
N_r(X) \\ 
 & & \hskip-15mm + \ tr 
\bigg(\int_{-\infty}^\infty f(x)^r \,dx\bigg)^{\frac{1+r}{1-r}}
\int_{-\infty}^\infty f(x)^{r-2} f'(x)^2 \, dx \\ 
 & & \hskip-15mm + \ 
\mbox{o}(t).
\ene

Using this representation, we are going to test the inequality \eqref{main} 
for $\alpha=1$, when it becomes
$$ 
N_r(X_t) \geq N_r(X) + t N_r(Z). 
$$
Comparing the linear terms in front of $t$ and using 
$N_r(Z) = 2 \pi r^{\frac{1}{r-1}}$, we would be led to a Nash-type inequality 
\be\label{5.1}
r \bigg(\int_{-\infty}^\infty f(x)^r \, dx \bigg)^{\frac{1+r}{1-r}} 
\int_{-\infty}^\infty f(x)^{r-2} f'(x)^2 \, dx 
 \, \geq \, 2\pi\, r^{\frac{1}{r-1}},
\en
holding already without too restrictive conditions 
(e.g., for all $C^1$-smooth $f>0$).

Now, let us take $f(x) = B e^{-\frac{|x|^p}{p}}$ with $p \geq 2$, where $B$
is a normalizing constant, i.e.,
$B^{-1} = 2\, p^{\frac{1}{p}-1}\, \Gamma(\frac{1}{p})$. In this case,
$$ 
\int_{-\infty}^\infty f(x)^r \, dx = 
B^r \int_{-\infty}^\infty e^{-\frac{r |x|^p}{p}} \, dx = 
B^{r-1} \frac{1}{r^{1/p}}, 
$$
so that
$$ 
\bigg(\int_{-\infty}^\infty f(x)^r \, dx \bigg)^{\frac{1+r}{1-r}} = \,
B^{-(r+1)} \, r^{\frac{r+1}{p(r-1)}}. 
$$
Similarly,
\begin{eqnarray*}
\int_{-\infty}^\infty f'(x)^2 f(x)^{r-2} \, dx 
 & = & 
B^r \int_{-\infty}^\infty |x|^{2(p-1)} e^{-\frac{r|x|^p}{p}} \, dx \\ 
 & = & 
2 B^r \, \Big(\frac{p}{r}\Big)^{\frac{2p-1}{p}}\, 
\frac{1}{p} \Gamma\Big(2 - \frac{1}{p}\Big),
\end{eqnarray*}
and thus the left-hand side in \eqref{5.1} is equal to
$$
\frac{2r}{B}\, r^{\frac{r+1}{p(r-1)}}
\Big(\frac{p}{r}\Big)^{\frac{2p-1}{p}} \frac{1}{p}\, 
\Gamma\Big(2 - \frac{1}{p}\Big) 
 =
4\, \Gamma\Big(\frac{1}{p}\Big)\, \Gamma\Big(2 - \frac{1}{p}\Big) \,
r^{\frac{2r}{p(r-1)} - 1}.
$$
Hence, inequality \eqref{5.1} says that
\be\label{5.2}
2 \pi \, \leq \,
4\, r^{-\frac{r(p-2)}{p(r-1)}} \, 
\Gamma\Big(\frac{1}{p}\Big)\, \Gamma\Big(2 - \frac{1}{p}\Big). 
\en

We claim that it cannot be true for all $p>2$ sufficiently
close to $2$ (i.e., when $X$ itself is almost standard normal). 
To see this, denote by $G(p)$ the right-hand side of \eqref{5.2} and note that
there is equality at $p=2$. So, let us look at the derivative and show that
$G'(2) < 0$, i.e., $H'(1/2) > 0$ for $H(x) = \log G(1/x)$. Indeed,
$$
H'(x) = \frac{2r}{r-1}\,\log r + \frac{\Gamma'(x)}{\Gamma(x)} - 
\frac{\Gamma'(2-x)}{\Gamma(2-x)}.
$$
From the fundamental relation $\Gamma(x+1) = x\Gamma(x)$, it follows that
$\Gamma'(x+1) = \Gamma(x) + x\Gamma'(x)$, so 
$\Gamma'(3/2) = \Gamma(1/2) + \frac{1}{2}\,\Gamma'(1/2)$, while
$\Gamma(3/2) = \frac{1}{2}\,\Gamma(1/2)$. Hence,
$$
H'(1/2) = \frac{2r}{r-1}\,\log r - 2 > 0.
$$
We may conclude that the entropy power inequality for $N_r$ of any order $r>1$
does not hold in general, even when one of the variable is Gaussian.

For another, less direct argument, one may return to \eqref{5.1} and rewrite
it as a homogeneous inequality
\begin{eqnarray*}
r \bigg(\int_\R f(x)^r \, dx \bigg)^{\frac{1+r}{1-r}} 
\int_\R f(x)^{r-2} f'(x)^2 \, dx 
 \\ \geq  2\pi\, r^{\frac{1}{r-1}}
\bigg(\int_\R f(x)\,dx\bigg)^{\frac{2r}{1-r}}.
\end{eqnarray*}
After the change $f=u^{\frac{2}{r}}$, it takes the form of the Nash-type
inequality
\be\label{5.3}
\bigg(\int_\R u(x)^2 \, dx\bigg)^{\frac{r+1}{r-1}} \, \leq \,
K_r \int_\R u'(x)^2 \, dx\, 
\bigg( \int_\R u(x)^{\frac{2}{r}}\,dx\bigg)^{\frac{2r}{r-1}}
\en
with $K_r = \frac{2}{\pi r^{\frac{r}{r-1}}}$. 
In fact, the Nash inequality in $\R^n$ asserts that
$$ 
\bigg(\int_{\R^n} \!u(x)^2 dx \bigg)^{1+\frac{2}{n}} \leq
C_n \int_{\R^n} |\nabla u(x)|^2 dx\, 
\bigg(\int_{\R^n} \!u(x) dx\bigg)^{\frac{4}{n}} 
$$
with sharp constant given by
$$ 
C_n = \Big(1 + \frac{2}{n}\Big)\, 
\Gamma \left( \frac{n}{2} + 2 \right)^{\frac{2}{n}} \frac{1}{\pi j_{n/2}^2}
$$
(cf. \cite{C-L}, \cite{B2}). Here $j_{\frac{n}{2}}$ denotes the smallest 
positive zero of the Bessel function $J_{\frac{n}{2}}$ of order $\frac{n}{2}$. 
In dimension $n=1$, one has 
$J_{\frac{1}{2}}(x) = \sqrt{\frac{2}{\pi x}}\,\sin(x)$ 
(cf. \cite{W}, p.\,54, eq. (3)), thus $j_{\frac{1}{2}} = \pi$. 
Hence the sharp Nash inequality in dimension 1 reads
$$ 
\bigg( \int u(x)^2 \, dx \bigg)^{3} \, \leq \frac{27}{16\, \pi^2} \,
 \int u'(x)^2 \, dx\, \bigg(\int u(x) \, dx\bigg)^{4},
$$
which is the same as \eqref{5.3} for $r=2$, however, with a larger constant.
Hence, as we have already seen, inequality \eqref{main} cannot be true
for $\alpha=1$ and $r=2$. Let us notice that the Nash inequality with 
the asymptotically sharp constant $2/(\pi e n)$ can be deduced from 
the classical EPI \eqref{EPI} (cf. \cite{T}).

For the parameter $r=2$, routine computations also provide a counterexample 
in the case where both $X$ and $Y$ have the beta distribution with density 
$f(x) = \frac{3}{4}\,(1-x^2)$, $|x|<1$ (sometimes called 
a $q$-Gaussian distribution).

\section{Lower bound on the optimal exponent}\label{6}

One may also provide a simple lower bound on the optimal exponent 
$\alpha = \alpha_{opt}$ that satisfies the inequality
$$
N_r(X+Y)^\alpha \geq N_r(X)^\alpha + N_r(Y)^\alpha
$$
for all independent random vectors $X$ and $Y$. Together with the upper bound 
of Theorem \ref{main-thm} and the counterexample in Section \ref{5}, we have:

\begin{prop}
One has
$$ 
\alpha_{opt} \in 
\Big[\min \Big\{1,\frac{\log 2}{2}\,\frac{r-1}{\log(\frac{r+1}{2})}\Big\},
\frac{r+1}{2}\,\Big]. 
$$
\end{prop}

\begin{proof}
For the remaining lower bound, let $X$ and $Y$ be independent and uniformly 
distributed on $[0,1]$, in which case $N_r(X) = N_r(Y) = 1$. The sum $X+Y$ 
has the triangle density $(f * g)(x) = x$ on $[0,1]$ and $(f * g)(x) = 2-x$ 
on $[1,2]$. Hence,
$$ 
\int (f * g)(x)^r \, dx = 
\int_0^1 x^r \, dx + \int_1^2 (2-x)^r \, dx = \frac{2}{r+1}. 
$$
Thus
$$ 
N_r(X+Y) = \left( \frac{r+1}{2} \right)^{\frac{2}{r-1}}. 
$$
Since 
$N_r(X+Y)^{\alpha_{opt}} \geq N_r(X)^{\alpha_{opt}} + N_r(Y)^{\alpha_{opt}}$, 
we deduce that 
$\left(\frac{r+1}{2} \right)^{\frac{2 \alpha_{opt}}{r-1}} \geq 2$,
which is the required statement.
\end{proof}

Let us stress that, if $X$ and $Y$ are independent real valued random variables 
with $N_r^{\alpha}(X + Y) = N_r^{\alpha}(X) + N_r^{\alpha}(Y)$, then drawing 
vectors ${\bf X} = (X_1, \dots, X_n)$  and ${\bf Y} = (Y_1, \dots, Y_n)$
with i.i.d. $X_i \sim X$ and $Y_i \sim Y$, we have
$$ 
N_r^{\alpha}({\bf X} + {\bf Y}) = N_r^{\alpha}({\bf X}) + N_r^{\alpha}({\bf Y}). 
$$
Hence, via this tensorization argument, there is no hope to improve $\alpha$ 
in higher dimension.

\section{Monotonicity and the CLT}\label{7}

Since the entropy power inequality (1) is closely related to the monotonicity
of the entropy along rescaled convolutions, let us make a remark,
restricting ourselves to the dimension $n=1$. Given an i.i.d. sequence 
of random variables $X,X_1,X_2,\dots$ with mean zero and variance one, 
the entropies $h(Z_k)$ of the normalized sums
$$
Z_k = \frac{X_1 + \dots + X_k}{\sqrt{k}}
$$
are known to be non-decreasing for growing $k$ and approaching the entropy $h(Z)$
of a standard normal random variable $Z$, cf. 
\cite{A-B-B-N}, \cite{Ba}, \cite{M-B}.
Since the monotonicity follows from (1), although for the subsequence 
$k = 2^l$ only, and since we have the more general inequality (3),
one may naturally wonder whether such a property extends to the R\'enyi's
entropies. This turns out to be false in general. If the 6-th moment $\E X^6$ 
is finite and $h_r(Z_{k_0})$ is finite for some $k_0$, a careful application 
of Edgeworth expansions yields an asymptotic representation
$$
\Delta_k(r) = h_r(Z) - h_r(Z_k) = B_r k^{-1} + C_r k^{-2} + o(k^{-2})
$$
with constant
$$
B_r = \frac{1}{4r}\, \bigg[\,
\frac{2-r}{3}\,\gamma_3^2 + \frac{r-1}{2}\,\gamma_4\bigg],
$$
where $\gamma_3 = \E X^3$ and $\gamma_4 = \E X^4 -3$ 
(the 3-rd and 4-th cumulants of $X$), and some constant $C_r \in \R$ 
(involving the cumulants of $X$ up to order 6). In the limit case $r=1$, 
such a representation, quantifying the entropic central limit theorem, was 
derived in \cite{B-C-G}. As for the values $r>1$, first suppose that 
$\gamma_3 \neq 0$. When $r$ is sufficiently close to 1, then $B_r>0$, 
so that $\Delta_k(r)$ is an eventually decreasing sequence like for $r=1$. 
More precisely, this is true for all $r>1$, whenever 
$\gamma_4 \geq \frac{2}{3}\,\gamma_3^2$. But, if 
$\gamma_4 < \frac{2}{3}\,\gamma_3^2$, then $B_r < 0$ for all 
$r > r_0 = (4\gamma_3^2 - 3\gamma_4)/(2\gamma_3^2 - 3\gamma_4)$, hence  
$\Delta_k(r)$ becomes an eventually increasing sequence. In that case, 
necessarily
$$
h_r(Z_k) > h_r(Z) \quad {\rm for \ all} \ k \ {\rm large \ enough},
$$
which is impossible in the Shannon case $r=1$. This also shows that
$\Delta_k(r)$ may not serve as distance.

If $\gamma_3 = 0$ (as in the situation of symmetric distributions), 
the constant is simplified to
$$
B_r = \frac{r-1}{8r}\, \gamma_4.
$$
Both cases, $\gamma_4>0$ or $\gamma_4<0$, are possible, and one
can make a similar conclusion as before for the whole range $r>1$.
We refer an interested reader to \cite{B-M} for more details.

\section*{Acknowledgment}
The authors would like to thank Eric Carlen, Eshed Ram and Igal Sason 
for reading the manuscript and for their valuable comments.
They are also grateful to both referees. In particular, one of them emphasized a
dimension-free character of the optimal value of $\alpha$, and the other one
raised the problem
of the monotonicity of the R\'enyi entropy in the central limit theorem.

\end{document}